\documentclass[conference]{IEEEtran}

\usepackage{latexsym, amssymb,graphicx, amsmath, subfigure,threeparttable}

\newtheorem{Theorem}{Theorem}
\newtheorem{Lemma}{Lemma}

\newtheorem{Proposition}{Proposition}

\newcommand{\remove}[1]{}

\begin{document}
%
\title{On  Optimal Ternary Locally Repairable Codes}

\author{\IEEEauthorblockN{Jie Hao\IEEEauthorrefmark{1}, Shu-Tao Xia\IEEEauthorrefmark{1}, and Bin Chen\IEEEauthorrefmark{1}\IEEEauthorrefmark{2}}
\IEEEauthorblockA{ \IEEEauthorrefmark{1}Graduate School at Shenzhen, Tsinghua University, Shenzhen, China\\
\IEEEauthorrefmark{2}School of Mathematical Sciences, South China Normal University, Guangzhou, China\\
Email: j-hao13@mails.tsinghua.edu.cn, xiast@sz.tsinghua.edu.cn, binchen14scnu@m.scnu.edu.cn}
}

\maketitle

\begin{abstract}
In an $[n,k,d]$ linear code, a code symbol  is said to have locality $r$ if it can be repaired  by accessing at most $r$ other code symbols.
For an $(n,k,r)$ \emph{locally repairable code} (LRC), the  minimum distance satisfies the well-known Singleton-like bound
$d\le n-k-\lceil k/r\rceil +2$.  In this paper, we study optimal ternary LRCs meeting this Singleton-like bound
 by employing a  parity-check matrix approach.
It is proved that there are only $8$  classes of possible parameters with which optimal ternary LRCs exist. Moreover, we obtain explicit constructions of optimal ternary LRCs for all these $8$ classes of parameters, where the minimum distance could only  be 2, 3, 4, 5 and 6.
\end{abstract}

\section{Introduction}
Recently, linear codes  with locality properties \cite{Gopalan,Papil} have attracted a lot of interest since their desirable applications in distributed storage systems.
Let $\mathbb{F}_q$ be a finite field with size $q$.
For a $q$-ary $[n,k,d]$ linear code with length $n$, dimension $k$ and minimum distance $d$, a code symbol with \emph{locality} $r$ means it can be repaired from at most $r$ other code symbols. In distributed storage systems, $r \ll k$ indicates low repair cost of a failed storage node.
An $(n,k,r)$ LRC is a $q$-ary $[n,k]$ linear code with locality $r$ for its code symbols. For an $(n,k,r)$ LRC  with locality for information symbols, Gopalan \emph{et al.} \cite{Gopalan} proved the following Singleton-like bound:
\begin{equation} \label{singleton-like-bound}
 d \leq n -k - \left\lceil {k}/{r} \right\rceil + 2,
\end{equation}
which reduces to the classical Singleton bound when $r=k$. For LRCs with locality for all code  symbols, Tamo \emph{et al.} \cite{Tamo-RS} gave another proof of the bound (\ref{singleton-like-bound}) by using results from graph theory. By a parity-check matrix approach, Hao and Xia \cite{Hao} analyzed $(n,k,r)$ LRCs with all symbol locality, where the bound (\ref{singleton-like-bound}) could also be naturally derived.

Linear codes with small field size are of special interest.  Many works have prosposed  constructions of optimal LRCs meeting the Singleton-like bound (\ref{singleton-like-bound}) over a relatively small field size.
Tamo \emph{et al.} \cite{Tamo-RS} proposed an elegant construction of optimal LRCs for $q \ge n+1$ by using polynomial methods. They further proposed optimal cyclic LRCs for $q \ge n+1$ in \cite{Tamo-cyclic}.
Ernvall \emph{et al.} proposed LRCs over a small alphabet in \cite{Ernvall}. Hao and Xia \cite{binaryLRC} proposed a class of optimal LRCs with $q \ge r-1, q = 2^m$ and $d=4$.
For optimal binary LRCs attaining the Singleton-like bound (\ref{singleton-like-bound}), Hao and Xia found all the possible four classes of optimal binary LRCs \cite{binaryLRC}.

Suppose $1\le k\le n-1$, $1\le r\le k-1$, and $\mathcal{C}$ is a $q$-ary $(n,k,r)$ LRC with all symbol locality throughout this paper. It is well known that nontrivial binary MDS codes meeting the Singleton bound do not exist, and the only binary MDS codes are binary $[n,1,n]$ and $[n,n-1,2]$ codes.
As for the ternary MDS codes, besides the trivial ternary $[n,1,n]$ and $[n,n-1,2]$ linear codes, the only possible one is the $[4,2,3]$ code. This could be seen from the following result.
\begin{Lemma}[\cite{MacWilliams}] \label{lemma-ternary-mds}
Let $\mathcal{C}$ be a $q$-ary $[n,k,d]$ MDS code. If $k \ge 2$, then $q \ge n-k+1$; If $k \le n-2$, then $q \ge k+1$.
\end{Lemma}

For $q=3$, since $k \le n-2$, $k \le q-1 = 2$; since $k=2$, $n \le q+k-1 = 4$. Hence the only possible nontrivial parameters are $n=4,k=2,d=3$.
\smallskip

In this paper, we study the constructions of optimal ternary LRCs  meeting the Singleton-like bound (\ref{singleton-like-bound}) and obtain the following main result.
\begin{Theorem}
\label{th-trivial}
Let $r\ge 1$, $k>r$ and $d\ge 2$. There are 8 classes of optimal ternary $(n,k,r)$ LRCs meeting the Singleton-like bound, whose parameters  are respectively
\begin{itemize}
\item $(k+\lceil k/r\rceil,k, r)$, $d=2$, $k>r\ge 1$;

\item $(13-g, 10-g, 8-g)$, $0 \le g \le 4$, $d=3$;

\item $(2k+2,k,1)$, $d=4$, $k\ge 2$;

\item $(8,2,1)$, $d=6$;

\item $(n,r+1,r)$, $d= n-k, 3\le k \le 6, 3 \le n-k \le 6,  n \le 12$;

\item $(4l,3l-2,3)$, $d=4$, $l\ge 3$;

\item $(3l,2l-1,2)$, $d=3$, $l\ge 3$;

\item $(12,5,2)$, $d=6$.

\end{itemize}
Except these 8 classes of LRCs, there is no other optimal ternary  $(n,k,r)$ LRC with $d=n-k-\lceil k/r\rceil +2$.
\end{Theorem}

Moreover, for each of these $8$ classes of possible parameters, we propose explicit constructions.
The rest of the paper is organized as follows. In Section II, some preliminaries on parity-check matrix approaches are presented. In Section III, we try to determine all the possible parameters that the optimal ternary LRCs can have, where the most complicated case is put in Section IV for clearance. Section V concludes the paper.

\section{Preliminaries}
For a $q$-ary $[n,k,d]$ linear code with length $n$, dimension $k$ and minimum distance $d$, the classical Singleton bound \cite{MacWilliams} says that $d\le n-k+1$ and the codes meeting it are called maximum distance separable (MDS) codes.
The \emph{support} of a vector is the set of coordinates of its non-zero components. If a coordinate is in the support of a vector, it is said to \emph{be covered} by the vector.
Let $A\otimes B $ be the Kronecker product of matrices. Let $I_m$ denote the $m \times m$ identity matrix.

Hao and Xia \cite{Hao} set up a new characterization of LRCs from the view of parity-check matrices, and then give an alternative proof of the bound (1). Let us  briefly  recall some of their results.

Let $\mathcal{C}$ be a $q$-ary $(n,k,r)$ LRC with minimum distance $d$.
In \cite{Hao}, $n-k$ parity-check equations are selected from the dual code $\mathcal{C}^\bot$ to form a full rank parity-check matrix $H$,
which is divided into two parts $(H_1^T, H_2^T)^T$. The rows in the upper part $H_1$, or \emph{locality-rows}, cover all coordinates and
ensure locality, while the rows in the lower part $H_2$ impact the minimum distance. Firstly, for the first coordinate, select a parity-check
equation with weight at most $r+1$ to cover it; then, for the first uncovered coordinate, select another parity-check equation with weight at
most $r+1$ to cover it; repeating the procedure iteratively, all coordinates are covered and $H_1$ is constructed. Let $l$ be the number of
rows of $H_1$ (or the number of locality-rows), then \cite{Hao}
\begin{eqnarray}
\label{krl}
\lceil k/r\rceil\le \lceil n/(r+1)\rceil\le l\le n-k.
\end{eqnarray}
Finally, some $n-k-l$ independent parity-check equations are selected to form $H_2$, and the construction of $H$ completes.

By deleting a fixed row of $H$ and all the columns it covered, we obtain a submatrix, say $H'$. It is clear that $d$ is upper bounded by the
minimum distance of the linear code with parity-check $H'$. Performing the above deleting procedure step by step for any fixed
$\left\lceil{k}/{r}\right\rceil-1$ locality-rows in $H_1$, we have
\begin{equation} \label{general-bound-d}
d \le \min_{ 1 \leq t \leq \left\lceil\frac{k}{r}\right\rceil-1} d^{(q)}_{\rm opt}\big (n  - t(r+1),k-tr\big),
\end{equation}
where $d_{\text{opt}}^{(q)}(n^*,k^*)$ is the largest possible minimum distance of a $q$-ary $[n^*,k^*]$ linear code.
When $t = \left\lceil{k}/{r}\right\rceil-1$ and invoking the classical Singleton bound $d_{\text{opt}}^{(q)} (n  - t(r+1),k-tr)\le n-k-t+1$,  the bound (\ref{general-bound-d}) reduces to the Singleton-like bound (\ref{singleton-like-bound}). When $\mathcal{C}$ is optimal, the following lemmas hold.

\begin{Lemma} [\cite{binaryLRC}] \label{lemma-binary-mds-lrc}
Let $\mathcal{C}$ be a $q$-ary $(n,k,r)$ LRC with $d=n-k-\lceil{k}/{r}\rceil+2$ and $H$ be its parity-check matrix described in the above procedures. Let $H'$ be the $m'\times n'$ matrix obtained from $H$ by deleting  any fixed  $\left\lceil{k}/{r}\right\rceil-1$ locality-rows and all the columns whose coordinates are covered by the supports  of these $\left\lceil{k}/{r}\right\rceil-1$ locality-rows. Then $H'$ has full rank and the $[n',k',d']$ linear code $\mathcal{C}'$ with the parity-check matrix $H'$ is a $q$-ary MDS code.
\end{Lemma}

\begin{Lemma}[\cite{Hao}]\label{necessary-condition}
For an $(n,k,r)$ LRC with $d = n-k-  \left\lceil{k}/{r}\right\rceil + 2$, suppose $r \mid k$, then $(r+1) \mid n$ and the supports of the locality-rows in the parity-check matrix must be pairwise disjoint, and each has weight exactly $r+1$.
\end{Lemma}

\section{Optimal Ternary LRCs  Meeting the Singleton-like Bound}

In this section, we will determine the parameters of  optimal ternary LRCs meeting the Singleton-like bound (\ref{singleton-like-bound}). It is proved that there exist only $8$ classes of possible parameters. For each class of parameter, we give explicit constructions.

The analysis procedure is similar to \cite{binaryLRC}, where all the binary optimal LRCs are enumerated. However, the ternary case is much more complicated, in which many new techniques are used to determine all the possible parameters.

Let $\mathcal{C}$ be an optimal ternary $(n,k,r)$ LRC with $d=n-k-\lceil{k}/{r}\rceil+2$ and $H$ be its parity-check matrix described in the premilaries.
The next result follows from Lemma \ref{lemma-binary-mds-lrc}.
\begin{Proposition} \label{co1}
Let $H'$ be the $m'\times n'$ matrix obtained from $H$ by deleting  any fixed  $\left\lceil{k}/{r}\right\rceil-1$ locality-rows and all the columns they covered. Then $H'$ has to be a full rank parity-check matrix of a ternary $[n',n'-1,2]\;(n'\ge 2)$ or $[4,2,3]$ or $[n',1,n'] \;(n'\ge 3)$ linear code.
\end{Proposition}

Now, we discuss each of these three classes of ternary LRCs meeting the Singleton-like bound.

\subsection{$H'$ Corresponds to a Ternary $[n',n'-1,2]\;(n'\ge 2)$ Code}
Clearly, $H'$ is a row vector with all the entries being $1$ or $2$. Moreover, since $\lceil k/r\rceil -1<l$ which is the number of locality-rows, $H'$ has a locality-row with weight at most $r+1$. This implies that $H'$ has to be a row vector with length at most $r+1$, or $n'\le r+1$. Since $n'\ge 2$, $d'=d=2$ and $n=k+\lceil k/r\rceil$. Hence, $\mathcal{C}$ must be a ternary $[k+\lceil k/r\rceil,k,2]$ linear code with locality $r$. Moreover, by (\ref{krl}), $\lceil k/r\rceil = l = n-k$, which implies that $H$ consists of only locality-rows.

If $r\mid k$, then $n=(r+1)k/r$, $n-k=k/r=l$, and all $k/r$ rows of $H$ must have uniform weight $r+1$ and pairwise disjoint supports.
Then, in the sense of the equivalence, the corresponding optimal ternary $(k+k/r,k,r)$ LRC must have the following parity-check matrix.
\begin{equation}
\label{h31}
H= \Big(I_{\frac{k}{r}} \otimes (\underbrace{1,1,\ldots,1}_{r+1})\Big)_{\frac{k}{r}\times \frac{(r+1)k}{r}}\;,
\end{equation}
e.g., if $n=9,k=6,r=2$, its parity-check matrix is
\begin{equation*}
H=
\left(
  \begin{array}{ccccccccc}
    1 & 1 & 1 & 0 & 0 & 0 & 0 & 0 & 0 \\
    0 & 0 & 0 & 1 & 1 & 1 & 0 & 0 & 0 \\
    0 & 0 & 0 & 0 & 0 & 0 & 1 & 1 & 1 \\
  \end{array}
\right).
\end{equation*}

If $r\nmid k$, then $r\ge 2$. Let $k=sr+t$, where $1\le t\le r-1$, then $\lceil \frac{k}{r}\rceil = s+1$, $n=k+\lceil \frac{k}{r}\rceil=(r+1)\lceil \frac{k}{r}\rceil-(r-t)$, where $1\le r-t\le r-1$. Let $\hat H$ be a  $\lceil \frac{k}{r}\rceil\times (r+1)\lceil \frac{k}{r}\rceil$ matrix in (\ref{h31}), where $\frac{k}{r}$ is changed to $\lceil \frac{k}{r}\rceil$.
\begin{eqnarray}
&&\mbox{$H$ is a $\lceil \frac{k}{r}\rceil\times (k+\lceil \frac{k}{r}\rceil)$ matrix obtained from $\hat H$}\nonumber\\
&&\qquad \mbox{by deleting any $r-t$ columns of $\hat H$, such}\nonumber\\
&&\qquad \mbox{that at least 1 row of $H$ has weight $r+1$;}
\qquad\label{h32} \\
&&\mbox{$\underline H$ is obtained from $H$ by substituting at most $r-t$}\nonumber\\
&&\qquad \mbox{0's of $H$ to 1's or 2's, such that the weight of each}\nonumber\\
&&\qquad \mbox{row of $\underline H$ is at most $r+1$.\nonumber}
\end{eqnarray}
Then, in the sense of the equivalence, every $(k+\lceil k/r\rceil,k,r)$ LRC with $d=2$ must have parity-check matrix as $H$ or $\underline H$, e.g., if $n=10,k=7,r=3$, its parity-check matrix is
\begin{equation*}
H=
\left(
  \begin{array}{ccccccccccc}
   1 & 1 & 1 & 1 & 0 & 0 & 0 & 0 & 0 & 0 \\
   0 & 0 & 0 & 0 & 1 & 1 & 1 & 1 & 0 & 0 \\
   0 & 0 & 0 & \underline{0} & 0 & 0 & \underline{0} & \underline{0} & 1 & 1 \\
  \end{array}
\right),
\end{equation*}
where any one or two of the three zeros with underline can be substituted to $1$ or $2$, and $\underline H$ is thus obtained.

Combining the above analysis in this subsection, we have the following lemma.
\begin{Lemma} \label{lemma-n-n-1-2}
When $H'$ corresponds to a ternary $[n',n'-1,2]$ $(n'\ge 2)$ code, optimal LRCs must have parameters $(k+\lceil k/r\rceil,$ $k,r,d=2)$. The parity-check matrices in (\ref{h31}) and (\ref{h32}) give respectively the constructions of $r\mid k$ and $r\nmid k$.
\end{Lemma}

\subsection{$H'$ Corresponds to a Ternary $[4,2,3]$ Code}

 $H'$ is a full rank parity-check matrix of a $[4,2,3]$ linear code. In the meaning of equivalence, $H'$ has to be
   \begin{equation}
H'=
\left(
  \begin{array}{ccccccccc}
    0 & 1 & 1 & 1 \\
    1 & 0 & 1 & 2
  \end{array}
\right).
\end{equation}
Then $m' =n - k - \left \lceil {k}/{r}\right  \rceil +1  = 2$, $n' =4$, $d'=d=3$. Thus $n = k+ \left \lceil {k}/{r}\right  \rceil +1$.  By (\ref{krl}), $\lceil k/r\rceil \le  l  \le n-k = \left \lceil {k}/{r}\right  \rceil +1$. Hence, the number of locality-rows $ l = \lceil k/r\rceil$ or $ l = \lceil k/r\rceil+1$.

\textbf{Case 1:} $ l = \lceil k/r\rceil$. When we obtain $H'$ by deleting $\left \lceil {k}/{r}\right  \rceil-1$ locality-rows, the remaining two rows must contain a locality-row which covers all the remaining 4 coordinates. However, the $[4,2,3]$ MDS codes with generator matrix $H'$ has weight distribution $A_0 = 1, A_2=0, A_3=8, A_4=0$, which implies that there does not exist a row which covers all the remaining 4 coordinates. This leads to a contradiction.

\medskip
\textbf{Case 2:} $ l = \lceil k/r\rceil +1$. Since $l=n-k$, all the rows in $H$ are locality-rows. Clearly, the two locality-rows corresponding to $H'$ intersect on at least two coordinates. Since we could delete arbitrary $\lceil k/r\rceil -1$ locality-rows, any two locality-rows intersect on at least two coordinates. Consider the first row and the last row of $H$, where the first $\lceil k/r\rceil -1$ rows will be deleted and $H'$ is left. They intersect on at least two coordinates which could not appear in $H'$. This implies that the last row of $H$ has weight at least 5, thus
\begin{eqnarray}
r\ge 4. \label{rge4}
\end{eqnarray}
Since the support of locality-rows are intersected, by Lemma \ref{necessary-condition}, $r \nmid k$. Let $k=sr+t$,
where $1\le t\le r-1$. Let $\gamma$ be the number of the columns covered by the supports of the deleted $\left \lceil {k}/{r}\right  \rceil-1$ locality-rows. Then \begin{eqnarray*}
&\gamma = n - n' = k + \left \lceil {k}/{r}\right  \rceil -3 \le (\lceil k/r\rceil -1)(r+1),&\\
\mbox{or}\quad
&k-\lceil k/r\rceil r+r\le 2.&
\end{eqnarray*}
By substituting $k= sr+t$, we have that $t = 1$ or $t = 2$.

  If $t=1$, $\gamma=k+\lceil k/r\rceil-3=s(r+1)-1$. Since we have delete $\lceil k/r\rceil -1= s$ locality-rows each of which has weight at most $r+1$, these deleted $s$ rows must intersect on exactly one coordinate. Since that any two locality-rows intersect on at least two coordinates, $\lceil k/r\rceil -1= s =1$, and this row has weight $r$. Since the $\lceil k/r\rceil -1= 1$ deleted row could be arbitrarily chosen, every locality-row has weight $r$, i.e., the code has locality $r-1$. This contradicts with the fact that the code has locality $r$.

  If $t=2$, $\gamma=k+\lceil k/r\rceil-3=s(r+1)$. Since we have deleted $\lceil k/r\rceil -1= s$ locality-rows each of which has weight at most $r+1$, these deleted $s$ rows must have disjoint support and have uniform weight exactly $r+1$. Since any two locality-rows intersect on at least two coordinates, we have $s=1$. Thus, $k=r+2$, $n=r+5$, $l=n-k=d=3$, which implies that the columns of $H$ have to be pairwise independent. Hence, the number of such columns $n\le \frac{q^{3}-1}{q-1}=13$. By (\ref{rge4}), $n=13,12,11,10,9$. In fact, they are a ternary Hamming code and its shortened versions, which lead to the following results.

\begin{Lemma} \label{lemma-423}
When $H'$ corresponds to a ternary $[4,2,3]$ code, optimal LRCs must have parameters $(13,10,8)$ or $(12,9,7)$ or $(11,8,6)$ or $(10,7,5)$ or $(9,6,4)$, all with $d=3$. By puncturing respectively the first $0,1,2,3,4$ columns from
\begin{equation}
\label{13-10-8}
\left(
  \begin{array}{ccccccccccccc}
  1 & 1 & 1 & 1 & 1 & 1 & 1 & 1 & 1 & 0 & 0 & 0 & 0 \\
  1 & 1 & 2 & 2 & 2 & 1 & 0 & 0 & 0 & 0 & 1 & 1 & 1 \\
  1 & 2 & 1 & 2 & 0 & 0 & 1 & 2 & 0 & 1 & 0 & 1 & 2
  \end{array}
\right),
\end{equation}
constructions of parity-check matrices are obtained.
\end{Lemma}

\section{$H'$ Corresponds to a $[n',1,n']\;(n'\ge 3)$ Code}

$H'$ is a full rank parity-check matrix of a ternary $[n',1,n']$ $(n'\ge 3)$ linear code. Then $H'$ is an $(n'-1)\times n'$ matrix, and $n' =d'=d= n - k - \left \lceil {k}/{r}\right  \rceil +2$. In Lemma \ref{lemma-binary-mds-lrc}, let $\gamma$ is the number of the columns covered by the supports of the deleted $\left \lceil {k}/{r}\right  \rceil-1$ locality-rows. Then
\begin{eqnarray}
&&\gamma = n - n' = k + \left \lceil {k}/{r}\right  \rceil -2\le (\lceil k/r\rceil -1)(r+1),\nonumber\\
\mbox{i.e.,}&&\qquad k-\lceil k/r\rceil r+r\le 1. \label{kr32}
\end{eqnarray}

\subsection{The Case of $r \mid k$}
If $r \mid k$, (\ref{kr32}) implies $r=1$ and $\gamma=2k-2$.
By Lemma \ref{necessary-condition}, $n = 2 l$. For any fixed $k-2$ locality-rows of $H$, let $H^*$ be obtained from $H$ by deleting these rows and all the columns they covered, where $H^*=H$ when $k =2$. Let $\mathcal{C}^*$ be the ternary $[n^*,k^*,d^*]$ linear code with the parity-check matrix $H^*$. Then $n^* = 2(l-k+2)$, $k^* = 2$ and $d^* = d = 2(l-k+1)$.
By the Plotkin bound with $q=3$ and $M^*=3^{k^*}$\cite{MacWilliams}
\begin{equation}\label{plotkin}
  d^* \le \frac{2n^*M^*}{3(M^*-1)},
\end{equation}
we have $l-k \le 2.$
By $n'=n-\gamma=2(l-k+1)\ge 3$, $l-k\ge 1$. Hence, we have $l=k+1$ or $l=k+2$.

For the case of $l=k+1$, $\mathcal{C}$ must have parameters $n=2k+2, k=k, r=1, d=4$, the following parity-check matrix gives optimal constructions.
\begin{equation}\label{construction-d-4-1}
    H = \left(
          \begin{array}{c}
            \quad I_{k+1} \;\otimes (1 \ 1)\\
            \hline
            \underbrace{(1 \cdots 1)}_{k+1} \otimes \;(0\  1)
             \\
          \end{array}
        \right)_{(k+2)\times (2k+2)},
\end{equation}
e.g., if $n=6,k=2,r=1$, it is
\begin{equation*}
H = \left(
  \begin{array}{cccccc}
    1 & 1 & 0 & 0 & 0 & 0  \\
    0 & 0 & 1 & 1 & 0 & 0  \\
    0 & 0 & 0 & 0 & 1 & 1  \\
    0 & 1 & 0 & 1 & 0 & 1
  \end{array}
\right).
\end{equation*}

For the case of $l=k+2$, $\mathcal{C}$ must have $n=2k+4, k=k, r=1, d=6$. The next result is needed for further studies.

\begin{Lemma}\label{lemma-l}
Let $H$ be an $(l+u) \times l(r+1)$ matrix where the first $l$ locality-rows have uniform weight $r+1$ and their supports are pairwise disjoint. If any 4 columns of $H$ are linearly independent, then
\begin{equation}
  l \le \frac{q^{u}-1}{(q-1) \cdot \binom{r+1}{2}}.
\end{equation}
\end{Lemma}
\begin{proof}
By the definition of $H$, each column has weight at most $u+1$, where the uppermost nonzero entry lies in the locality-row. Given a locality-row with weight $r+1$, it covers $r+1$ columns, any two of which could result in a column vector whose first $l$ entries are zeros by eliminating their uppermost nonzero entries. Thus, $l$ locality-rows result in $l{r+1\choose 2}$ column vectors in total, and in each of these column vectors, the first $l$ entries are zeros while the remaining $u$ entries are arbitrary. It is easy to see that the condition that any 4 columns of $H$ are linearly independent implies that all the above $l{r+1\choose 2}$ column vectors have to be pairwise independent, i.e., one vector is not the multiple of another one. Since the maximum number of the pairwise independent $u$-dimensional vectors is $\frac{q^u-1}{q-1}$,
\begin{equation}
  l \cdot \binom{r+1}{2} \le \frac{q^{u}-1}{q-1},
\end{equation}
which finishes the proof.
\end{proof}

In this case of $l=k+2$, $u=n-k-l=2$. By Lemma \ref{lemma-l}, $l=k+2 \le \frac{3^2-1}{3-1}=4.$ Hence $k=2$ since $k>r=1$, or $n=8, k=2, r=1, d=6$, and the parity-check matrix is
\begin{equation} \label{construction-8-2-6}
H = \left(
  \begin{array}{cccccccc}
    1 & 1 & 0 & 0 & 0 & 0 & 0 & 0  \\
    0 & 0 & 1 & 1 & 0 & 0 & 0 & 0  \\
    0 & 0 & 0 & 0 & 1 & 1 & 0 & 0  \\
    0 & 0 & 0 & 0 & 0 & 0 & 1 & 1  \\
    0 & 1 & 0 & 1 & 0 & 1 & 0 & 0  \\
    0 & 2 & 0 & 1 & 0 & 0 & 0 & 1
  \end{array}
\right).
\end{equation}

Combining these, we have the following lemma.
\begin{Lemma} \label{lemma-n-1-n-case1}
When $H'$ corresponds to a ternary $[n',1,n']$ $(n'\ge 3)$ code, if $r\mid k$, optimal LRCs must have parameters $(n=2k+2, k=k, r=1, d=4)$ or $(n=8, k=2, r=1, d=6)$.  The codes with the parity-check matrices in (\ref{construction-d-4-1}) and (\ref{construction-8-2-6}) give the optimal constructions, respectively.
\end{Lemma}

\subsection{The Case of $r \nmid k$}

If $r \nmid k$, then $r\ge 2$. Let $k = sr+t$, where $1\le t\le r-1$. By substituting $k= sr+t$ into (\ref{kr32}), we have $t=1$, which implies that $k=sr+1$ and $\gamma=k+\lceil k/r\rceil-2=s(r+1)$.
Since we delete $\lceil k/r\rceil -1= s\ge 1$ locality-rows each of which has weight at most $r+1$, these deleted $s$ rows must have disjoint support and have uniform weight exactly $r+1$. Since the $s$ locality-rows are arbitrarily chosen, we conclude that all locality-rows of $H$ have uniform weight $r+1$. Next, we will break it into two cases: $s=1$ and $s\ge 2$.

{\bf Case \textsl{s}}$\,\mathbf{= 1}$: $k=r+1$, $\lceil k/r\rceil=2$, $d=n-k=n'\ge 3$. By $r\ge 2$, $k\ge 3$.
A $q$-ary $[n,k,d]$ code is called a \emph{near MDS} code if $d=d_1=n-k$ and $d_i=n-k+i$, $i=2,3,\ldots,k$, where $d_i$ denote the $i$-th \emph{generalized Hamming weight}  of the code \cite{Huffman}.
Let $\mathcal{C}^{\bot}$ be the dual code of $\mathcal{C}$ and $d_i^{\bot}$ be the $i$-th generalized Hamming weight of $\mathcal{C}^{\bot}$. Since all locality-rows of $H$ have uniform weight $r+1$, $d_1^{\bot} = r+1$. Since $\mathcal{C}$ has $d= n-k-\lceil k/r \rceil +2$, by the duality of the generalized Hamming weights, it follows that for $2 = \lceil k/r\rceil \le i \le n-k$,  $ d_i^{\bot} = k+i$ \cite{Huffman}. Hence $\mathcal{C}^{\bot}$ is a near MDS code, which indicates that $\mathcal{C}$ is also a near MDS code \cite{nearMDS}.
Since $k\ge 3$, $n-k\ge 3$ and $q=3$, it follows that  \cite[Theorem 3.5]{nearMDS}
$$3\le k\le 2q=6, \; 3\le n-k\le 2q=6,\;\mbox{ and }\;n\le 12.$$ Thus the possible ternary near MDS codes are respectively [12,6,6], [11,6,5], [11,5,6], [10,6,4], [10,5,5],[10,4,6], [9,6,3], [9,5,4], [9,4,5], [9,3,6], [8,5,3],[8,4,4], [8,3,5],  [7,4,3], [7,3,4],  [6,3,3], 16 ternary $[n,k,d]$ linear codes in total.

It can be easily verified, e.g., by using the MAGMA software, that there do exist 16 optimal LRCs with all the above parameters, where the locality $r = k-1$. For example, the ternary $[12,6,6]$ LRC with $r=5$ can be obtained by extending the  ternary $[11,6,5]$ quadratic residue code, whose parity-check matrix could be
\begin{equation}
H = \left(
  \begin{array}{cccccccccccc}
    1 & 2 & 2 & 1 & 2 & 1 & 0 & 0 & 0 & 0 & 0 & 0 \\
    0 & 0 & 0 & 0 & 0 & 0 & 1 & 2 & 2 & 2 & 1 & 1 \\
    0 & 0 & 1 & 1 & 1 & 1 & 0 & 0 & 0 & 1 & 1 & 0 \\
    0 & 1 & 0 & 2 & 1 & 0 & 0 & 0 & 0 & 1 & 2 & 2  \\
    0 & 0 & 0 & 2 & 0 & 1 & 0 & 1 & 0 & 2 & 1 & 2  \\
    0 & 0 & 0 & 1 & 2 & 1 & 0 & 0 & 1 & 0 & 2 & 2
  \end{array}
\right).
\end{equation}
All the 15 other optimal LRCs can be obtained from puncturing or shortening this $[12,6,6]$ code. Moreover, their localities also satisfy $r = k-1$.

\begin{Lemma} \label{lemma-n-1-n-case2}
When $H'$ corresponds to an $[n',1,n']$ $(n'\ge 3)$ ternary linear code, if $r\nmid k$ and $k=r+1$, there are exact 16 $[n,k,d]$ optimal LRCs, where $3\le k \le 6$, $3 \le n-k \le 6$,  $6\le n \le 12$, and $d= n-k$. Moreover, all these LRCs are  ternary near MDS codes.
\end{Lemma}

\medskip
{\bf Case \textsl{s}}$\,\mathbf{\ge 2}$:  Since the deleted $s$ locality-rows are arbitrarily chosen, we conclude that all locality-rows of $H$ are pairwise disjoint, which implies $(r+1)\mid n$ or $n = l(r+1)$, then $n' = n-\gamma=(l-s)(r+1)$. In Lemma \ref{lemma-binary-mds-lrc}, for any fixed $s-1$ locality-rows of $H$, let $H^*$ be obtained from $H$ by deleting these rows and all the corresponding columns. Let $\mathcal{C}^*$ be the ternary $[n^*,k^*,d^*]$ linear code with the parity-check matrix $H^*$. Then
\begin{eqnarray*}
&&n^* = (l-s+1)(r+1),\quad k^* = r+1, \\
&&d^* = n-k-\left \lceil {k}/{r}\right \rceil+2 = (l-s)(r+1).
\end{eqnarray*}
By the Plotkin bound (\ref{plotkin}),
\begin{equation}
    (l-s)(r+1) \le   2 \cdot \frac{(l-s+1)(r+1)\cdot 3^{r+1}}{3 \cdot (3^{r+1}-1)},
\end{equation}
or $l-s \le {2 \cdot 3^{r}}/{(3^{r}-1)}$.
Since $r\ge 2$, $l-s = 1 \mbox{ or } l-s=2.$

\medskip
\subsubsection{$l-s = 1$ }
$n^* = 2(r+1), \ k^* = r+1, \ d^* =  r+1$.

Let $M_q(n,d)$ denote the maximum number of codewords in a $q$-ary linear code with length $n$ and minimum distance $d$. By shortening techniques, it is easy to see that \cite{MacWilliams}
\begin{equation}\label{M-q}
M_q(n,d) \le q M_q(n\!-\!1,d) \mbox{ or } M_3(n,d) \le 3 M_3(n\!-\!1,d).
\end{equation}
When $d$ is even, by (\ref{plotkin}),
$M_3({3d}/{2} - 1,d) \le 3d/2.$
Hence,
\begin{equation}\label{M-3d/2}
  M_3({3d}/{2},d) \le 3 \cdot M_3({3d}/{2} - 1,d) \le {9d}/{2}.
\end{equation}

If $r$ is odd, then $d^* = r+1$ is even. By (\ref{M-q}) and (\ref{M-3d/2}),
\begin{eqnarray}
M_3(2d^*,d^*) &\le& 3 \cdot M_3(2d^*-1,d^*)  \nonumber\\
    &\le& \cdots \nonumber\\
    &\le& 3^{{d^*}/{2}}\cdot M_3({3d^*}/{2},d^*) \nonumber\\
    &\le& 3^{{d^*}/{2}} \cdot {9d^*}/{2}.
\end{eqnarray}
Hence, $3^{r+1} \le 3^{\frac{r+1}{2}} \cdot \frac{9(r+1)}{2},$ which implies that $r \le 5$. Since $r$ is odd, we have $r=3$ or $r=5$.
When $r=3$, $\mathcal{C}$ must have parameters $n=l(r+1)=4l$, $k=sr+1=(l-1)*3+1=3l-2$, $r=3$, $d=4$, where $l\ge 3$ by $s\ge 2$, and  its parity-check matrix can be
\begin{equation}\label{construction-d-4-2}
    H = \left(
          \begin{array}{c}
            I_l \otimes (1 \ 1 \ 1 \ 1)\\
            \hline
            \underbrace{(1 \ 1 \cdots 1)}_{l} \otimes \left(
                                                        \begin{array}{cccc}
                                                          0 & 0 & 1 & 1 \\
                                                          0 & 1 & 0 & 1 \\
                                                        \end{array}
                                                      \right)
             \\
          \end{array}
        \right)_{(l+2)\times 4l},
\end{equation}
e.g., if $n=12,k=7,r=3$, it is
\begin{equation*}
H = \left(
  \begin{array}{cccccccccccc}
    1 & 1 & 1 & 1 & 0 & 0 & 0 & 0 & 0 & 0 & 0 & 0  \\
    0 & 0 & 0 & 0 & 1 & 1 & 1 & 1 & 0 & 0 & 0 & 0  \\
    0 & 0 & 0 & 0 & 0 & 0 & 0 & 0 & 1 & 1 & 1 & 1  \\
    0 & 0 & 1 & 1 & 0 & 0 & 1 & 1 & 0 & 0 & 1 & 1  \\
    0 & 1 & 0 & 1 & 0 & 1 & 0 & 1 & 0 & 1 & 0 & 1
  \end{array}
\right).
\end{equation*}
When $r=5$, $\mathcal{C}$ must have parameters $n=l(r+1)=6l$, $k=sr+1=(l-1)*5+1=5l-4$, $r=5$, $d=6$. By Lemma \ref{lemma-l}, $l \le \frac{3^4-1}{(3-1){5+1\choose 2}}=\frac{8}{3}$. Since $l\ge 3$ by $s\ge 2$, there is no optimal LRC in this case.

\smallskip

If $r$ is even, then $d^* = r+1$ is odd. By the bound (\ref{plotkin}),
$
  M_3({(3d^*-1)}/{2} ,d^*) \le 3d^*.
$
By (\ref{M-q}),
\begin{eqnarray}
M_3(2d^*,d^*) &\le& 3 \cdot M_3(2d^*-1,d^*)  \nonumber\\
    &\le& \cdots \nonumber\\
    &\le& 3^{\frac{d^*+1}{2}}\cdot M_3((3d^*-1)/2,d^*) \nonumber\\
    &\le& 3^{\frac{d^*+1}{2}} \cdot 3d^*.
\end{eqnarray}
Hence, $3^{r+1} \le 3^{\frac{r+2}{2}} \cdot 3d,$ which implies that $r \le 5$. Since $r$ is even, we have $r=2$ or $r=4$.
When $r=2$, $\mathcal{C}$ must have parameters $n=l(r+1)=3l$, $k=sr+1=(l-1)*2+1=2l-1$, $r=2$, $d=3$, where $l\ge 3$ by $s\ge 2$, and  its parity-check matrix can be
\begin{equation}\label{construction-d-4-3}
    H = \left(
          \begin{array}{c}
            I_l \otimes (1 \ 1 \ 1 )\\
            \hline
            \underbrace{(1 \ 1 \cdots 1)}_{l} \otimes (0 \ 1 \ 2 )
             \\
          \end{array}
        \right)_{(l+1)\times 3l},
\end{equation}
e.g., if $n=9,k=5,r=2$, it is
\begin{equation*}
H = \left(
  \begin{array}{ccccccccc}
    1 & 1 & 1  & 0 & 0 & 0 & 0 & 0 & 0  \\
    0 & 0 & 0  & 1 & 1 & 1 & 0 & 0 & 0  \\
    0 & 0 & 0  & 0 & 0 & 0 & 1 & 1 & 1   \\
    0 & 1 & 2  & 0 & 1 & 2 & 0 & 1 & 2
  \end{array}
\right).
\end{equation*}
When $r=4$, $\mathcal{C}$ must have parameters $n=l(r+1)=5l$, $k=sr+1=(l-1)*4+1=4l-3$, $r=4$, $d=5$.
By Lemma \ref{lemma-l}, $l \le 1$. So there is no optimal LRC in this case.

\bigskip
\subsubsection{$l-s = 2$ }
$n^* = 3(r+1), \ k^* = r+1, \ d^* =  2(r+1)$.

By (\ref{M-3d/2}), we have
$3^{r+1} \le \frac{9}{2} \cdot 2(r+1),$
or $r \le 2$. Since $r\ge 2$, $r=2$, which implies that $\mathcal{C}$ has parameters $n=3l$, $k=2l-3$, $r=2$, $d=6$.
By Lemma \ref{lemma-l},
$
  l \le  \frac{3^3-1}{(3-1)\cdot \binom{3}{2}}  = \frac{13}{3}.
$
Thus $l=4$ by $s\ge2$. Hence, only the $(12,5,2,6)$ code is possible. The following parity-check matrix gives an optimal $(12,5,2)$ LRC with $d = 6$.
\begin{equation}
\label{construction-d-4-4}
H = \left(
  \begin{array}{cccccccccccc}
    1 & 1 & 1 & 0 & 0 & 0 & 0 & 0 & 0 & 0 & 0 & 0 \\
    0 & 0 & 0 & 1 & 1 & 1 & 0 & 0 & 0 & 0 & 0 & 0   \\
    0 & 0 & 0 & 0 & 0 & 0 & 1 & 1 & 1 & 0 & 0 & 0  \\
    0 & 0 & 0 & 0 & 0 & 0 & 0 & 0 & 0 & 1 & 1 & 1   \\
    0 & 2 & 1 & 0 & 2 & 1 & 0 & 2 & 1 & 0 & 0 & 0   \\
    0 & 0 & 2 & 0 & 2 & 2 & 0 & 2 & 0 & 0 & 1 & 0   \\
    0 & 1 & 1 & 0 & 1 & 0 & 0 & 2 & 0 & 0 & 0 & 1
  \end{array}
\right).
\end{equation}

Combining these, we have the following lemma.
\begin{Lemma} \label{lemma-n-1-n-case3}
When $H'$ corresponds to a ternary $[n',1,n']$ $(n' \ge 3)$ code, if $r\nmid k$ and $k>r+1$, optimal LRCs must have parameters $(n=4l, k=3l-2, r = 3, d = 4)$ or $(n=3l, k=2l-1, r =2,d= 3)$ or  $(n=12, k=5, r=2, d=6)$, where $l\ge 3$.
The codes with the parity-check matrices (\ref{construction-d-4-2}), (\ref{construction-d-4-3}) and (\ref{construction-d-4-4}) give the optimal constructions, respectively.
\end{Lemma}

\bigskip
Combining Lemmas \ref{lemma-n-n-1-2}, \ref{lemma-423}, \ref{lemma-n-1-n-case1}, \ref{lemma-n-1-n-case2} and \ref{lemma-n-1-n-case3}, the main theorem of this paper follows.

\section{Conclusions}
In this paper, we study the constructions of optimal ternary LRCs based on a parity-check matrix approach. It is proved that there are only $8$ classes of possible parameters that optimal ternary LRCs can achieve. The minimum distance of optimal ternary LRCs
 can only  be 2, 3, 4, 5, 6. Moreover, we propose constructions of optimal ternary LRCs for all these $8$ possible classes of   parameters.

\section*{Acknowledgment}
This research is supported in part by the National Natural Science Foundation of China under grant No. 61371078, and the R\&D Program of Shenzhen under grant Nos. JCYJ20140509172959977, JSGG20150512162853495, ZDSYS20140509172959989, JCYJ20160331184440545.



%

\end{document}